\documentclass{article}
\usepackage{tikz,amsmath,amssymb,amsthm,setspace}
\usepackage{amsfonts,inputenc,graphics,txfonts}
\usepackage{algorithm}
\usepackage{algpseudocode}
\newtheorem{prop}{Proposition}
\newtheorem{cor}{Corollary}
\newtheorem{claim}{Claim}
\newtheorem{lemma}{Lemma}
\newtheorem{theorem}{Theorem}
\newtheorem{definition}{Definition}

\author{%
Jakkepalli Pavan Kumar \\
Department of Computer Science and Engineering\\
National Institute of Technology\\
Warangal, Telangana, India}

\author{
	Jakkepalli Pavan Kumar, P. Venkata Subba Reddy \\
	Department of Computer Science and Engineering\\
	National Institute of Technology\\
	Warangal, Telangana, India}

\title{Algorithmic Aspects of Secure Connected Domination in Graphs}

\date{}

\begin{document}
\maketitle
\begin{abstract}
	Let $G = (V,E)$ be a simple, undirected and connected graph. 
	A connected dominating set $S \subseteq V$ is a secure connected dominating set of $G$, if for each $ u \in V\setminus S$, there exists $v \in S$ such that $(u,v) \in E$ and the set $(S \setminus \lbrace v \rbrace) \cup \lbrace u \rbrace $ is a connected dominating set of $G$. The minimum size of a secure connected dominating set of $G$ denoted by $ \gamma_{sc} (G)$, is called the secure connected domination number of $G$. Given a graph $ G$ and a positive integer $ k,$ the Secure Connected Domination (SCDM) problem is to check whether $ G $ has a secure connected dominating set of size at most $ k.$ In this paper, we prove that the SCDM problem is NP-complete for doubly chordal graphs, a subclass of chordal graphs. We investigate the complexity of this problem for some subclasses of bipartite graphs namely, star convex bipartite, comb convex bipartite, chordal bipartite and chain graphs. The Minimum Secure Connected Dominating Set (MSCDS) problem is to find a secure connected dominating set of minimum size in the input graph. We propose a $ (\Delta(G)+1) $ - approximation algorithm for MSCDS, where $ \Delta(G) $ is the maximum degree of the input graph $ G $ and prove that MSCDS cannot be approximated within $ (1 − \epsilon) \ln(\vert V \vert ) $ for any $ \epsilon > 0 $ unless $ NP \subseteq DTIME(\vert V \vert^{ O(\log \log \vert V \vert )}) $ even for bipartite graphs. Finally, we show that the MSCDS is APX-complete for graphs with $\Delta(G)=4.$
\end{abstract}

\section{Introduction}
Throughout this paper, all graphs $G=(V,E)$ should be finite, simple (i.e., without self-loops and multiple edges), undirected and connected with vertex set $V$ and edge set $E$. The (\textit{open}) \textit{neighborhood} of a vertex $v \in V$ is the set $N(v) = \{u \in V$ $\vert$ $(u, v) \in E\}$. If $X \subseteq V$, then the open neighborhood of $X$ is the set $N(X) = \cup_{v \in X} N(v)$. The closed neighborhood of $X$ is $N[X] = X \cup N(X)$.  The \textit{degree} of a vertex $v$ is $\vert N(v)\vert$ and is denoted by $d(v)$. If $d(v) = 1$, then $v$ is called a \textit{pendant vertex} of $G$ and the support vertex of a pendant vertex $ v $ is the unique vertex $ s $ such that $ (v,s) \in E $. The maximum degree of vertices in $ V $ is denoted by $ \Delta(G) $. For a graph $G,$ and a set $S \subseteq V,$ the subgraph of $G$ \textit{induced} by $S$ is defined as $G[S]=(S,E_S)$, where $E_S=\{(x,y)\in E$ $:$ $x,y \in S\}$. If $G[S]$, where $S\subseteq V$, is a complete subgraph of $G$, then it is called a \textit{clique} of $G$. A set $S \subseteq V$ is an \textit{independent set} if $G[S]$ has no edge. A \textit{split graph} is a graph in which the vertices can be partitioned into a clique and an independent set. For undefined terminology and notations we refer to \cite{west}.\par
A set $S \subseteq V$ is a \textit{dominating set} (DS) in $G$ if for every $u \in V \setminus S$, there exists $v \in S$ such that $(u,v) \in E$, i.e., $N[S] = V$. The minimum size of dominating set in $G$ is called the \textit{domination number} of $G$ and is denoted by $\gamma(G)$.	A set $S \subseteq V$ is a \textit{connected dominating set} (CDS) of $G$ if $G[S]$ is connected and every vertex not in $S$ is adjacent to a vertex in $S$. The minimum size of CDS in $G$ is called the \textit{connected domination number} of $G$ and is denoted by $\gamma_c(G)$. The study of domination and related problems is one of the fastest growing areas in graph theory.
The study of literature on various domination parameters in graphs has been surveyed and outlined in \cite{Haynes,67}. An important domination parameter called secure domination has been introduced by E.J. Cockayne et al. in \cite{pog}. A set $S \subseteq V$ is a \textit{secure dominating set} (SDS) of $G$ if, for each vertex $u \in V \setminus S$, there exists a neighboring vertex $v$ of $u$ in $S$ such that $(S \setminus \{v\}) \cup \{u\}$ is a dominating set of $G$ (in which case $v$ is said to \textit{defend} $u$). The decision version of secure domination problem is known to be NP-complete for general graphs \cite{dev} and remains NP-complete even for various restricted families of graphs such as bipartite, doubly chordal and split graphs \cite{osd,sdsnp}. Recently, H. Wang et al. \cite{sdsnp} obtained some approximation results related to secure domination. A CDS $S$ of $G$ is called a \textit{secure connected dominating set} (SCDS) in $G$ if, for each $ u \in V \setminus S$, there exists $ v \in S$ such that  $(u,v) \in E$ and $(S \setminus \lbrace v \rbrace ) \cup \lbrace u \rbrace $ is a CDS in $G$ (in which case $v$ is said to be $ S $-\textit{defender}). The \textit{secure connected domination number} of graph $G$ is the minimum size of a SCDS, and is denoted by $\gamma_{sc}(G)$  \cite{sc1}. Given a graph $ G$ and a positive integer $ k,$ the Secure Connected Domination (SCDM) problem is to check whether $ G $ has a SCDS of size at most $ k.$ It is known that SCDM is NP-complete for bipartite graphs and split graphs, whereas it is linear time solvable for block graphs and threshold graphs \cite{akce}. The Minimum Secure Connected Dominating Set (MSCDS) problem is to find a SCDS of minimum size in the input graph.\par
\paragraph{Preliminaries} In a graph $G$, a vertex $v$ is \textit{simplicial} if its closed neighborhood $N_G[v]$ induces a complete subgraph of $G$. An ordering $\{v_1, v_2, \ldots, v_n\}$ of the vertices of $V$ is a \textit{perfect elimination ordering} (PEO), if $v_i$ is a simplicial of the induced subgraph $G_i = G[\{v_i, v_{i+1}, \ldots, v_n\}]$ for every $i$, $1 \le i \le n$. A graph $G$ is \textit{chordal} if and only if $G$ admits a PEO. A vertex $u \in N[v]$ is a \textit{maximum neighbor} of $v$ in $G$ if $N[w] \subseteq N[u]$ holds for each $w \in N[v]$. A vertex $v \in V$ is called \textit{doubly simplicial} if it is a simplicial vertex and has a maximum neighbor. An ordering $\{v_1, v_2, \ldots, v_n\}$ of the vertices of $V$ is a \textit{doubly perfect elimination ordering} (DPEO) of $G$ if $v_i$ is a doubly simplicial vertex of the induced subgraph $G_i = G[\{v_i, v_{i+1}, \ldots, v_n\}]$ for every $i$, $1 \le i \le n$. A graph $G$ is \textit{doubly chordal} if and only if it has a DPEO \cite{doubly}. Alternatively, doubly chordal graphs are chordal and dually chordal graphs. An undirected graph is a \textit{tree} if it is connected and cycle-free. A \textit{star} is a tree $T=(A, F)$, where $A=\{a_0, a_1, \ldots, a_n\}$ and $F=\{(a_0, a_i) \vert 1 \le i \le n\}$. A \textit{comb} is a tree $T=(A,F)$, where $A=\{a_1, a_2, \ldots, a_{2n}\}$ and $F=\{(a_i,a_{i+1})$ $\vert$ $1 \le i \le n-1\}$ $\cup$ $\{(a_i,a_{n+i})$ $\vert$ $1 \le i \le n\}$. A bipartite graph $G = (A, B, E)$ is called \textit{tree convex bipartite graph} if there is an associated tree $T=(A, F)$ such that for each vertex $b$ in $B$, its neighborhood $N_G(b)$ induces a subtree of $T$ \cite{jiang}. If $T$ is a star (or comb), then $G$ is called as \textit{star convex bipartite} (or \textit{comb convex bipartite}) graph. 
A graph $ G $ is \textit{chordal bipartite} if $ G $ is bipartite and each cycle of $ G $ of length greater than $ 4 $ has a chord. Alternatively, chordal bipartite graphs are weakly chordal and bipartite graphs.
\section{Complexity Results}
In this section, we show that the complexity of SCDM in doubly chordal, star convex bipartite, comb convex bipartite, and chordal bipartite graphs is NP-complete. Also, we prove that SCDM is linear time solvable in chain graphs, a subclass of bipartite graphs. The decision version of domination and secure connected domination problems are defined as follows.\\[4pt]
\noindent \textbf{Domination Decision Problem (DOMINATION)} \\ 
\textit{Instance:} A simple, undirected graph $G$ and a positive integer $k$.\\
\textit{Question:} Does there exist a dominating set of size at most $ k $ in $ G $?\\[9pt]
\textbf{Secure Connected Domination Problem (SCDM)} \\
\textit{Instance}: A simple, undirected and connected graph $G$ and a positive integer $l$.\\
\textit{Question}: Does there exist a SCDS of size at most $ l $ in $ G $? \\[3pt]
\noindent Domination decision problem for bipartite graphs has been proved as NP-complete \cite{bersto}. 
Let $ P(G) $ and
$ S(G) $ be the set of pendant and support vertices of $ G $, respectively.
\begin{prop}(\cite{sc1})\label{p1}
	Let $ G $ be a connected graph of order $ n \ge 3 $ and let $ X $ be a secure connected dominating set of $ G $. Then\\
	(i) $ P(G) \subseteq X $ and $ S(G) \subseteq X$,\\
	(ii) no vertex in $ P(G) \cup S(G) $ is an X-\textit{defender}.
\end{prop}
\subsection{Secure connected domination for doubly chordal graphs}\label{complexity}
\noindent To prove the NP-completeness of the SCDM for doubly chordal graphs we consider the following SET-COVER decision problem which has been proved as NP-complete \cite{karp}. \smallskip \\[3pt]
\textbf{Set Cover Decision Problem (SET-COVER)} \\ [6pt]
\textit{Instance:} A finite set $X$ of elements, a family of $m$ subsets of elements $C$ and a\\
 \hspace*{1.7cm}positive integer $k$.\\
\textit{Question:} Does there exist a subfamily of $k$ subsets $C^\prime$ whose union equals $X$? 
\begin{theorem}
	SCDM is NP-complete for doubly chordal graphs.
\end{theorem}
\begin{proof}
	Clearly, SCDM is in NP. If a set $S \subseteq V$, such that $|S| \le l$ is given as a witness to a yes instance then it can be verified in polynomial time that $S$ is a SCDS of $G$.\par
	Let $X=\{x_1$, $x_2,\ldots,x_n\}$, $C=\{C_1,C_2,\ldots,C_m\}$ be an instance of SET-COVER problem. We now construct an instance of SCDM from the given instance of SET-COVER as similar to the reduction in \cite{sdsnp} as follows. Construct a graph $G$ with the following vertices: (i) a vertex $x_i$ for each element $x_i \in X$, (ii) vertex $c_j$ for each subset $C_j \in C$ and let $C^*=\{c_j : 1 \le j \le m\}$ and (iii) two vertices $p$ and $q$.  Add the following edges in $G$: (i) if $x_i \in C_j$, then add edge $(x_i,c_j)$, where $1\le i \le n$ and $1 \le j \le m$, (ii) edges between every pair of vertices in the set $C^* \cup \{p\}$, (iii) edges between $x_i$ and $p$, where $1\le i \le n$ and (iv) edge between $p$ and $q$. Since $ G $ admits a DPEO $\{x_1, x_2, \ldots, x_n, c_1, c_2, \ldots, c_m, p,q\}$, it is a doubly chordal graph and the construction of $G$ can be accomplished in polynomial time.\par
	Now we show that the given instance of SET-COVER problem $<X,C>$ has a set cover of size at most $k$ if and only if the constructed graph $G$ has a SCDS of size at most $l=k+2$. Suppose $C^\prime \subseteq C$ is a set cover of $X$, with $\vert C^\prime \vert \le k$, then it is easy to verify that the set $$S= \{c_j : C_j \in C^\prime\} \cup \{p,q\}$$ is a SCDS of size at most $k+2$ in $G$.\par 
	Conversely, suppose $S \subseteq V$ is a SCDS of size at most $l=k+2$ in $G$. From Proposition \ref{p1}, it is clear that $\vert S \cap \{p,q\} \vert = 2$. Let $X^* = S \cap X$ and $S^*= S \cap \{c_j : 1\le j \le m\}$. If $\vert X^* \vert=0$, then we are done, that is respective subsets of vertices of $S^*$ form the solution for SET-COVER and clearly $\vert S^* \vert \le k$. Otherwise, since $X$ is an independent set, every vertex in $X^*$ can be replaced with its adjacent vertex in the set $C^*$ and size of the resultant set is at most $ k $. Therefore, there exists a set cover of size at most $k$. 
\end{proof}
\subsection{Secure connected domination for subclasses of bipartite graphs}
\begin{theorem}
	SCDM is NP-complete for star convex bipartite graphs.
\end{theorem}
\begin{proof}
	It is known that SCDM is in NP. We reduce the Domination problem for bipartite graphs to SCDM for star convex bipartite graphs as follows. The reduction is similar to the construction given in \cite{sdsnp}. Given an instance $G = (A, B, E)$ of Domination problem for bipartite graphs, where 
	$A = \{a_1, a_2, \ldots, a_p\}$ and $B =\{b_1, b_2, \ldots, b_q\}$, we construct an instance $G^\prime = (A^\prime, B^\prime, E^\prime)$ of SCDM, where $A^\prime = A$ $\cup$ $\{a_x, a_y\}$, $B^\prime = B$ $\cup$ $\{b_x, b_y\}$, and $E^\prime = E \cup \{(a_x,b_i) : 1 \le i \le q\} \cup \{(b_x,a_i) : 1 \le i \le p\} \cup \{(a_x,b_x),(a_x,b_y),(b_x,a_y)\}.$ It can be verified that $G^\prime$ is a star convex bipartite graph with its associated star $T=(A^\prime, F)$, where $F = \{(a_x,a_i)$ $\vert$ $1\le i \le p\}$ $\cup$ $\{(a_x,a_y)\}$. Note that the construction of graph $G^\prime$ can be done in polynomial time. \par
	Next we show that $G$ has a dominating set of size at most $k$ if and only if $G^\prime$ has a SCDS of size at most $l=k+4$. Suppose $D$ is a dominating set in $G$ of size at most $k$. Then it can be easily verified that the set $D$ $\cup$ $\{a_x, a_y, b_x, b_y\}$ is a SCDS in $G^\prime$ of size at most $k+4$. \par
	On the other hand, let $S$ be a SCDS in $G^\prime$ with $\vert S \vert \le l=k+4$. From Proposition \ref{p1}, it is clear that $\vert S \cap \{a_x,a_y,b_x,b_y\}\vert=4$. Let $S^*=S \setminus \{a_x, a_y, b_x, b_y\}$. Thus, $\vert S^* \vert \le k$. Since $ S $ is a SCDS and $ a_x$ and  $b_x $ are support vertices, for every vertex $a_i \in A$, for $1 \le i \le p$, $\vert S^* \cap N_{G^\prime}[a_i]\vert \ge 1$. Similarly, for every vertex $b_i \in B$, for $1 \le i \le q$, $\vert S^* \cap N_{G^\prime}[b_i]\vert \ge 1$. Therefore, $S^*$ is a dominating set of size at most $k$ in $G$.  
\end{proof}

\begin{theorem}
	SCDM is NP-complete for comb convex bipartite graphs.
\end{theorem}
\begin{proof}
	It is known that SCDM is in NP. To prove the NP-hardness of SCDM for comb convex bipartite graphs we reduce from Domination problem for bipartite graphs. Given an instance $G = (A, B, E)$ of Domination problem for bipartite graphs, where 
	$A = \{a_1, a_2, \ldots, a_p\}$ and $B =\{b_1, b_2, \ldots, b_q\}$, we construct an instance $G^\prime = (A^\prime, B^\prime, E^\prime)$ of SCDM, where $A^\prime = A$ $\cup$ $\{a_{p+1}^\prime, a_{p+2}^\prime, \ldots,$ $a_{2p}^\prime\}$ $\cup \{a_x, a_y\} $, $B^\prime = B$ $\cup$ $\{b_{p+1}^\prime, b_{p+2}^\prime, \ldots, b_{2p}^\prime\}$ $\cup$ $\{b_x\}$, and $E^\prime = E \cup \{(a_i^\prime, b_j) : p+1 \le i \le 2p, 1 \le j \le q\} \cup \{(a_i^\prime, b_i^\prime) : p+1 \le i \le 2p\} \cup \{(a_i, b_x) : 1 \le i \le p\} \cup \{(a_{p+i}^\prime, b_x) : 1 \le i \le p\} \cup \{(a_x,b_x),(a_y,b_x)\}$. It can be verified that $G^\prime$ is a comb convex bipartite graph with its associated comb $T=(A^\prime, F)$, with backbone $\{a_{p+1}^\prime,a_{p+2}^\prime, \ldots, a_{2p}^\prime, a_x\}$ and teeth $\{a_1,a_2,\ldots,a_p,a_y\}$. It can be noted that the construction of graph $G^\prime$ can be done in polynomial time. \par
	Next we show that $G$ has a dominating set of size at most $k$ if and only if $G^\prime$ has a SCDS of size at most $l=k+2p+3$. Suppose $D$ is a dominating set in $G$ of size at most $k$. Then it can be easily verified that the set $D \cup \{a_{p+i}^\prime, b_{p+i}^\prime : 1 \le i \le p\} \cup \{a_x, a_y, b_x\}$ is a SCDS in $G^\prime$ of size at most $k+2p+3$. \par
	Conversely, let $S$ be a SCDS of size at most $k+2p+3$ in $ G^\prime$. Let $A^*=\{a_{p+i}^\prime$ $:$ $1 \le i \le p\}$ and $B^*=\{b_{p+i}^\prime : 1 \le i \le p\}$. From Proposition \ref{p1}, it is clear that $\vert S \cap A^*\vert=p$, $\vert S \cap B^*\vert=p$ and $\vert S \cap \{a_x,a_y,b_x\} \vert =3$. Suppose $S^*=S \cap V$, then $\vert S^* \vert \le k$. Since $S$ is a SCDS of $G^\prime$, it can be easily verified that for every vertex $v \in A \cup B$, $N[v] \cap S^* \ne \emptyset$. Therefore, $S^*$ is a dominating set of size at most $k$.
\end{proof}
\noindent The following Vertex-Cover problem has been proved as NP-complete \cite{karp}, which will be used to show SCDM for chordal bipartite graphs as NP-complete.\\
A \textit{vertex cover} of an undirected graph $ G = (V, E) $ is a subset of vertices $ V^\prime \subseteq V$ such that if edge $ (u, v) \in E $, then either $ u \in V^\prime $ or $ v \in V^\prime$ or both.\\[5pt]
\noindent \textbf{Vertex Cover Decision Problem (Vertex-Cover)} \\ [6pt]
\textit{Instance:} A simple, undirected graph $G$ and a positive integer $k$.\\
\textit{Question:} Does there exist a vertex cover of size at most $ k $ in $ G$?

\begin{theorem}
	SCDM is NP-complete for chordal bipartite graphs.
\end{theorem}
\begin{proof}
	It is known that SCDM is in NP. To prove NP-hardness of SCDM for chordal bipartite graphs we reduce from Vertex-Cover. The reduction is similar to the construction given in \cite{miller}. Given an instance $ G=(V,E) $ of Vertex-Cover, where $\vert V \vert=n$ and $ \vert E \vert =m $, we construct an instance $ G^\prime =(V^\prime, E^\prime)$ of SCDM as follows.
	\begin{enumerate}
		\item Replace each vertex $ i \in V $ by a component $ G_i=(V_i, E_i): $\\
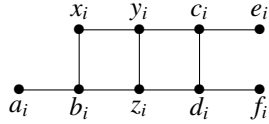
		\begin{figure}[H]
			\begin{center}
				\begin{tikzpicture}[scale=0.8]
				\node at (-1,0){\textbullet};\node at (-1,-0.3){$ a_i $};
				\node at (0,0){\textbullet};\node at (0,-0.3){$ b_i $};
				\node at (1,0){\textbullet};\node at (1,-0.3){$ z_i $};
				\node at (2,0){\textbullet};\node at (2,-0.3){$ d_i $};
				\node at (3,0){\textbullet};\node at (3,-0.3){$ f_i $};
				
				\node at (0,1){\textbullet};\node at (0,1.3){$ x_i $};
				\node at (1,1){\textbullet};\node at (1,1.3){$ y_i $};
				\node at (2,1){\textbullet};\node at (2,1.3){$ c_i $};
				\node at (3,1){\textbullet};\node at (3,1.3){$ e_i $};
				
				\draw (-1,0)--(0,0)--(1,0)--(2,0)--(3,0);
				\draw (0,1)--(1,1)--(2,1)--(3,1);
				\draw (0,0)--(0,1);
				\draw (1,0)--(1,1);
				\draw (2,0)--(2,1);
				\end{tikzpicture}
			\end{center}
		\end{figure}
		%
		%
		%
		%
		\item Replace each edge $(i,j) \in E $ by the following components $ G_{ij}=(V_{ij}, E_{ij}) $ (Figure (a)) and $ G_{ji}=(V_{ji}, E_{ji}) $ (Figure (b))
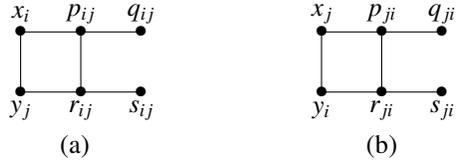
		\begin{figure}[H]
			\begin{center}
				\begin{tikzpicture}[scale=0.8]
				\node at (0,0){\textbullet};\node at (0,-0.3){$ y_j $};
				\node at (1,0){\textbullet};\node at (1,-0.3){$ r_{ij} $};
				\node at (2,0){\textbullet};\node at (2,-0.3){$ s_{ij} $};
				
				\node at (0,1){\textbullet};\node at (0,1.3){$ x_i $};
				\node at (1,1){\textbullet};\node at (1,1.3){$ p_{ij} $};
				\node at (2,1){\textbullet};\node at (2,1.3){$ q_{ij} $};
				
				\draw (0,0)--(1,0)--(2,0);
				\draw (0,1)--(1,1)--(2,1);
				\draw (0,0)--(0,1);
				\draw (1,0)--(1,1);
				
				\node at (0.9,-0.9) {(a)};
				
				
				
				\node at (5,0){\textbullet};\node at (5,-0.3){$ y_i $};
				\node at (6,0){\textbullet};\node at (6,-0.3){$ r_{ji} $};
				\node at (7,0){\textbullet};\node at (7,-0.3){$ s_{ji} $};
				
				\node at (5,1){\textbullet};\node at (5,1.3){$ x_j $};
				\node at (6,1){\textbullet};\node at (6,1.3){$ p_{ji} $};
				\node at (7,1){\textbullet};\node at (7,1.3){$ q_{ji} $};
				
				\draw (5,0)--(6,0)--(7,0);
				\draw (5,1)--(6,1)--(7,1);
				\draw (5,0)--(5,1);
				\draw (6,0)--(6,1);
				
				\node at (6,-0.9) {(b)};
				\end{tikzpicture}
			\end{center}
		\end{figure}
		Let $ X=\{x_i: i=1,\ldots,n\},  Y=\{y_i: i=1,\ldots,n\},$
		$ Z=\{z_i: i=1,\ldots,n\}$,  $K=X \cup Y \cup Z,$
		$ A=\{a_i,b_i,c_i,d_i,e_i,f_i: i=1,\ldots,n\}$, and $B=\{p_{ij},q_{ij},p_{ji},q_{ji},r_{ij},s_{ij},r_{ji},s_{ji}: (i,j)\in E\}$.
		\item Add two more additional vertices $ t $ and $ u$ such that $V^\prime=K \cup A \cup B \cup \{t, u\}$,\\
		$ E^\prime = \bigcup\limits_{i=1}^{n}E_i$ $\cup \bigcup\limits_{(i,j)\in E} (E_{ij} \cup E_{ji}) \cup \{(x_i,y_j),(z_i,y_j) : i=1,\ldots,n$ $\&$ $j=1,\ldots,n\} \cup  \{(x_i,u), (z_i,u)$, $(y_i,t) : i = 1, \ldots, n\} \cup \{(t,u)\}$.
	\end{enumerate} 
	Since $ V^\prime $ can be partitioned into two independent sets $ X \cup Z \cup \{a_i,c_i,f_i : i=1,\ldots,n\} \cup \{q_{ij},q_{ji},r_{ij},r_{ji} : (i,j)\in E\} \cup \{t\}$ and  $ Y \cup \{b_i,d_i,e_i : i=1,\ldots,n\} \cup \{p_{ij},p_{ji},s_{ij},s_{ji} : (i,j)\in E\} \cup \{u\}$, the constructed graph $ G^\prime $ is a bipartite graph. \par 
	\begin{figure}
		\begin{tikzpicture}
		\node at (-2.4,2.5){\textbullet};  \node at (-2.4,2.7){{\tiny 1}};
		\node at (-1.8,2.5){\textbullet};  \node at (-1.8,2.7){{\tiny 2}};
		\node at (-2.4,1.9){\textbullet};  \node at (-2.4,1.7){{\tiny 3}};
		\node at (-1.8,1.9){\textbullet};  \node at (-1.8,1.7){{\tiny 4}};
		\draw (-1.8,1.9)--(-2.4,1.9)--(-2.4,2.5)--(-1.8,2.5);
		\node at (-2,1.4){{\tiny Graph $ G $}};
		
		\node at (-1.2,2.1){$ \Rightarrow $};
		
		\node at(0,5) { \textbullet };	\node at(0,5.2) { {\tiny $ x_1 $}};
		\node at(1,5) { \textbullet };  \node at(1,5.2) {{\tiny  $ y_1 $}};
		\node at(2,5) { \textbullet };  \node at(2,5.2) { {\tiny $ c_1 $}};
		\node at(3,5) { \textbullet };  \node at(3,5.2) { {\tiny $ e_1 $}};
		\draw (0,5)--(1,5)--(2,5)--(3,5);
		
		\draw (1,5)..controls(4,2.5)..(7,1);
		\draw (1,5)--(8,4);
		\draw (1,5)..controls(2.5,2)..(8,0);
		\draw (1,5) to [bend right=27] (1,0);
		
		\node at(0,4) { \textbullet };  	\node at(0,3.8) {{\tiny $ b_1 $}};
		\node at(1,4) { \textbullet };		\node at(1,3.8) { {\tiny $ z_1 $} };
		\node at(2,4) { \textbullet };		\node at(2,3.8) { {\tiny $ d_1 $} };
		\node at(3,4) { \textbullet };		\node at(3,3.8) { {\tiny $ f_1 $} };
		\node at(-1,4) { \textbullet };		\node at(-1,3.8) { {\tiny $ a_1 $} };
		\draw (-1,4)--(0,4)--(1,4)--(2,4)--(3,4);
		
		\draw (0,5)--(0,4);
		\draw (1,5)--(1,4);
		\draw (2,5)--(2,4);
								
		\node at(7,5) { \textbullet }; \node at(7,5.2) { {\tiny $ x_2 $}};
		\node at(8,5) { \textbullet };\node at(8,5.2) {{\tiny  $ y_2 $}};
		\node at(9,5) { \textbullet };\node at(9,5.2) { {\tiny $ c_2 $}};
		\node at(10,5) { \textbullet };\node at(10,5.2) { {\tiny $ e_2 $}};
		\draw (7,5)--(8,5)--(9,5)--(10,5);
		
		\node at(7,4) { \textbullet };\node at(7,3.8) { {\tiny $ b_2$}};
		\node at(8,4) { \textbullet };\node at(8,3.8) { {\tiny $ z_2 $}};		
		\node at(9,4) { \textbullet };\node at(9,3.8) { {\tiny $ d_2 $}};
		\node at(10,4) { \textbullet };\node at(10,3.8) { {\tiny $ f_2 $}};
		\node at(6,4) { \textbullet };\node at(6,3.8) { {\tiny $ a_2 $}};
		\draw (6,4)--(7,4)--(8,4)--(9,4)--(10,4);

		\draw (7,5)--(7,4);
		\draw (8,5)--(8,4);
		\draw (9,5)--(9,4);
		
		\draw (8,5) to [bend left=20](1,4);
		\draw (8,5) to [bend left=10](0,1);
		\draw (8,5)--(7,1);
		\draw (8,5) to [bend left=14](1,0);
		\draw (8,5) to [bend left=15](8,0);
		
		\node at(0,0) { \textbullet };\node at(0,-0.2) { {\tiny $ b_3 $}};
		\node at(1,0) { \textbullet };\node at(1,-0.2) { {\tiny $ z_3 $}};
		\node at(2,0) { \textbullet };\node at(2,-0.2) { {\tiny $ d_3 $}};
		\node at(3,0) { \textbullet };\node at(3,-0.2) {{\tiny  $ f_3 $}};
		\node at(-1,0) { \textbullet };\node at(-1,-0.2) { {\tiny $ a_3 $}};
		\draw (-1,0)--(0,0)--(1,0)--(2,0)--(3,0);
		
		\node at(0,1) { \textbullet };\node at(-0.2,1.2) { {\tiny $ x_3 $}};
		\node at(1,1) { \textbullet };\node at(0.8,0.8) { {\tiny $ y_3 $}};
		\node at(2,1) { \textbullet };\node at(2.1,1.2) { {\tiny $ c_3 $}};
		\node at(3,1) { \textbullet };\node at(3,1.2) { {\tiny $ e_3 $}};
		\draw (0,1)--(1,1)--(2,1)--(3,1);	
		
		\draw (0,1)--(0,0);
		\draw (1,1)--(1,0);
		\draw (2,1)--(2,0);
		
		\draw (1,1) to [bend left=20] (1,4);
		\draw (1,1) to [bend left=7] (8,4);
		\draw (1,1) to [bend left=3] (7,5);
		\draw (1,1)--(8,0);
		
		\node at(7,0) { \textbullet };\node at(7,-0.2) { {\tiny $ b_4 $}};
		\node at(8,0) { \textbullet };\node at(8,-0.2) { {\tiny $ z_4 $}};
		\node at(9,0) { \textbullet };\node at(9,-0.2) { {\tiny $ d_4 $}};
		\node at(10,0) { \textbullet };\node at(10,-0.2) { {\tiny $ f_4 $}};
		\node at(6,0) { \textbullet };\node at(6,-0.2) { {\tiny $ a_4 $}};
		\draw (6,0)--(7,0)--(8,0)--(9,0)--(10,0);
			
		\node at(7,1) { \textbullet };\node at(6.6,1) { {\tiny $ x_4 $}};
		\node at(8,1) { \textbullet };\node at(7.8,1.25) { {\tiny $ y_4 $}};
		\node at(9,1) { \textbullet };\node at(9,1.2) { {\tiny $ c_4 $}};
		\node at(10,1) { \textbullet };\node at(10,1.2) { {\tiny $ e_4 $}};
		\draw (7,1)--(8,1)--(9,1)--(10,1);
		
		\draw (7,0)--(7,1);
		\draw (8,0)--(8,1);
		\draw (9,0)--(9,1);
		
		\draw (8,1)--(1,0);
		\draw (8,1) to [bend right=17](8,4);
		\draw (8,1)to [bend left=7](7,5);
		\draw (8,1)to [bend left=7](0,5);
		\draw (8,1)to [bend left=2](1,4);
		\draw (0,5) to [bend left=17] (8,5);
		\draw (0,5) to [bend left=27] (8,5);
		\draw (1,5) to [bend left=-17] (7,5);
		\draw (1,5) to [bend left=-27] (7,5);
		
		\node at (3.5,6.05){\textbullet};		\node at (3.5,5.9){{\tiny $ p_{12} $}};
		\node at (4.5,6.05){\textbullet};		\node at (4.5,5.9){{\tiny $ r_{12} $}};
		\node at (3.5,6.75){\textbullet};		\node at (3.5,7){{\tiny $ q_{12} $}};
		\node at (4.5,6.75){\textbullet};		\node at (4.5,7){{\tiny $ s_{12} $}};
		\draw (3.5,6.05)--(3.5,6.75);
		\draw (4.5,6.05)--(4.5,6.75);	
		
			\node at (3.5,4.18){\textbullet};		\node at (3.75,4.35){{\tiny $ r_{21} $}};
			\node at (4.5,4.18){\textbullet};		\node at (4.75,4.35){{\tiny $ p_{21} $}};
			\node at (3.5,3.58){\textbullet};		\node at (3.75,3.5){{\tiny $ s_{21} $}};
			\node at (4.5,3.58){\textbullet};		\node at (4.75,3.5){{\tiny $ q_{21} $}};
			\draw (3.5,4.18)--(3.5,3.58);
			\draw (4.5,4.18)--(4.5,3.58);
				
		\draw (0,5) to [bend left=11] (1,1);
		\draw (0,5) to [bend left=27] (1,1);
		\draw (1,5) to [bend right=11] (0,1);
		\draw (1,5) to [bend right=27] (0,1);
		
		\node at (-0.15,2.5){\textbullet};		\node at (-0.3,2.65){{\tiny $ p_{31} $}};
		\node at (0,3.25){\textbullet};			\node at (-0.2,3.45){{\tiny $ r_{31} $}};
		\node at (-0.75,2.5){\textbullet};		\node at (-0.9,2.65){{\tiny $ q_{31} $}};
		\node at (-0.55,3.25){\textbullet};		\node at (-0.75,3.45){{\tiny $ s_{31} $}};
		\draw (-0.15,2.5)--(-0.75,2.5);
		\draw (0,3.25)--(-0.55,3.25);

		\node at (1.15,2.5){\textbullet};		\node at (1.25,2.3){{\tiny $ r_{13} $}};
		\node at (1,3.25){\textbullet};			\node at (1.1,3.05){{\tiny $ p_{13} $}};
		\node at (1.75,2.5){\textbullet};		\node at (1.9,2.3){{\tiny $ s_{13} $}};
		\node at (1.6,3.25){\textbullet};		\node at (1.6,3.05){{\tiny $ q_{13} $}};
		\draw (1.15,2.5)--(1.75,2.5);
		\draw (1,3.25)--(1.6,3.25);
		
		\draw (0,1) to [bend left=17] (8,1);
		\draw (0,1) to [bend left=27] (8,1);
		\draw (1,1) to [bend left=-17] (7,1);
		\draw (1,1) to [bend left=-27] (7,1);
		
			\node at (3.5,2){\textbullet};		\node at (3.7,1.85){{\tiny $ p_{34} $}};
			\node at (4.5,2){\textbullet};		\node at (4.4,1.85){{\tiny $ r_{34} $}};
			\node at (3.5,2.6){\textbullet};	\node at (3.73,2.6){{\tiny $ q_{34} $}};
			\node at (4.5,2.6){\textbullet};	\node at (4.75,2.6){{\tiny $ s_{34} $}};
			\draw (3.5,2)--(3.5,2.6);
			\draw (4.5,2)--(4.5,2.6);

		\node at (3.5,0.2){\textbullet};		\node at (3.7,0.05){{\tiny $ r_{43} $}};
		\node at (4.5,0.2){\textbullet};		\node at (4.7,0.05){{\tiny $ p_{43} $}};
		\node at (3.5,-0.4){\textbullet};		\node at (3.7,-0.55){{\tiny $ s_{43} $}};
		\node at (4.5,-0.4){\textbullet};		\node at (4.7,-0.55){{\tiny $ q_{43} $}};
		\draw (3.5,0.2)--(3.5,-0.4);
		\draw (4.5,0.2)--(4.5,-0.4);
		
			\node at (3.5,-1){{\tiny Graph $ G^\prime $}};

		\node at (5.2,5.2){\textbullet};  \node at (5.2,5.4){{\tiny $ u $}};
		\draw (0,5) to [bend right=10](5.2,5.2);
		\draw (1,4)--(5.2,5.2);
		\draw (7,5) to [bend right=10](5.2,5.2);
		\draw (8,4)--(5.2,5.2);
		\draw (7,1) to [bend left=10](5.2,5.2);
		\draw (8,0)  to [bend left=10](5.2,5.2);
		\draw (0,1) to [bend left=-1](5.2,5.2);
		\draw (1,0)  to [bend left=02](5.2,5.2);
		
		\node at (4.3,0.9){\textbullet}; \node at (4.3,1.1){{\tiny$ t $}};
		\draw (4.3,0.9) to [bend left=15](1,1);
		\draw (4.3,0.9) to [bend left=5](1,5);
		\draw (4.3,0.9) to [bend right=15](8,1);
		\draw (4.3,0.9) to [bend right=10](8,5);
		
		\draw (5.2,5.2) to [bend left=15](4.3,0.9);
		\end{tikzpicture}
		\caption{Example construction of graph $ G^\prime$ from graph $ G $}\label{cbg294}
	\end{figure}
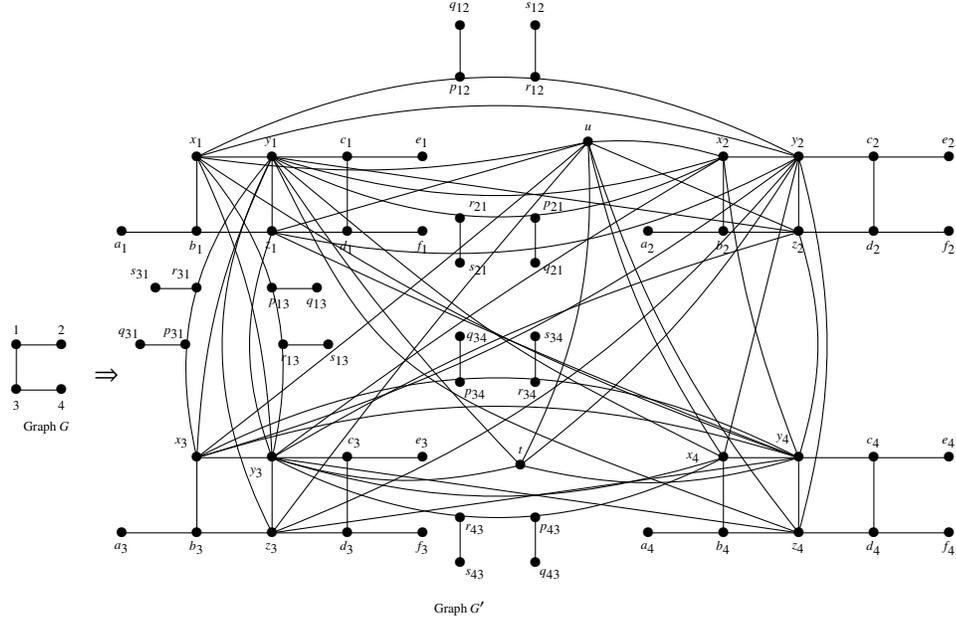
	Let $ C $ be a cycle in $ G^\prime $ of length greater than $ 4.$ If $ C $ is a cycle within a component $ G_i $ for some $ i $, then clearly it contains $ y_i $. Otherwise, if $ C $ is a cycle formed with vertices from more than one $ G_i $ component then it contains either edge $ (x_k, y_l) $ or $ (z_k,y_l) $.  
	 Therefore, each cycle of length greater than $ 4 $ contains at least one vertex $ y_i \in Y.$ If $ C $ contains exactly one $ y_i \in Y, $ (i) if $ C = G_i$  then $ (y_i,z_i) $ is a chord, (ii) if $ C $ contains $ u $  then $ (u, z_j) $ is a chord, and (iii) if $ C $ contains $ t $ then $ (y_i, z_j) $ is a chord. If $ C $ contains at least two vertices $ y_i, y_j $ from $ Y $ and (i) if $ C $ contains $ c_i $ or $ c_j $ then $ (y_i,z_i)$ or $ (y_j,z_j) $ is a chord, (ii) if $ C $ contains $ r_{ij} $ or $ r_{ji} $ then $ (y_i,c_j) $ is a chord, (iii) since vertices $ u $ and $t$ are adjacent to every vertex $ v^\prime \in X \cup Z $ and $u^\prime \in Y $ respectively, if $ C $ contains $ t $ or $ u $ then there exists a chord. Therefore, $ G^\prime$ is a chordal bipartite graph and can be constructed in polynomial time. An example construction of graph $ G^\prime $ from graph $G$ is illustrated in Figure \ref{cbg294}.\\[4pt]   
	We show that $ G $ has a vertex cover of size at most $ k $ if and only if $ G^\prime $ has a SCDS of size at most $ 7n+8m+k+2. $ Let $ VC $ be a vertex cover of $ G $ of size at most $ k $. Let $ S=\{a_i,b_i,c_i,d_i,e_i,f_i : i \in V\} \cup \{p_{ij},q_{ij},r_{ij},s_{ij},p_{ji},q_{ji},r_{ji},s_{ji} : (i,j) \in E\} \cup \{x_i,y_i : i \in VC\} \cup \{z_i : i \notin VC\}$ $ \cup$ $\{t,u\} $. It can be verified that $ S $ forms a SCDS of $ G^\prime $ and $\vert S \vert = 6n+8m+2k+(n-k)+2= 7n+8m+k+2.$ \\[4pt]
	\noindent Conversely, suppose $ S^\prime $ is a SCDS of size at most $ 7n+8m+k+2.$ 
	\begin{claim}\label{xiyi}
		If $ x_i \in S^\prime $ then without loss of generality, $ y_i \in S^\prime$ and vice versa.
	\end{claim}	 
	\begin{proof}[Proof of claim]
		Let $ x_i \in S^\prime $. Since $ S^\prime $ is a CDS, then it is true that either $ y_i \in S^\prime $ or $ z_i \in S^\prime.$ Then, take without loss of generality, $ y_i \in S^\prime. $ Analogously, if $ y_i \in S^\prime $, then either $ x_i \in S^\prime $ or $ z_i \in S^\prime.$  Then, take without loss of generality, $ x_i \in S^\prime.$
	\end{proof}
	\begin{claim}
		If $ S^\prime $ is a SCDS of $ G^\prime $ with $\vert  S^\prime \cap \{t,u\} \vert < 2 $ then there exists a SCDS of $ G^\prime $ with the same size and $ \vert  S^\prime \cap \{t,u\} \vert = 2 $.
	\end{claim}	 
	\begin{proof}[Proof of claim]
		Since $ X \cup Z \cup \{t\} $ and $ Y \cup \{u\} $ forms a complete bipartite subgraph in $ G^\prime $, if $ S^\prime $ is a SCDS of $ G^\prime $ and $ t,u \notin S^\prime $ then there exists two vertices $ v_1 \in S^\prime \cap Y $, $ v_2 \in S^\prime \cap (X \cup Z) $ such that $ (S^\prime \setminus \{v_1,v_2\}) \cup \{t, u\}$ is also a SCDS of $ G^\prime.$ With the similar argument, if $ t \notin S^\prime $ (or $ u\notin S^\prime $) then there exists a vertex $ v_1 \in S^\prime \cap Y $ (or $ v_2 \in S^\prime \cap X $) such that $ (S^\prime \setminus \{ v_1 \}) \cup \{t\} $ (or $ (S^\prime \setminus \{ v_2 \}) \cup \{u\} $) is a SCDS of $ G^\prime. $  Hence the claim.
	\end{proof}	
	\noindent Let $ S_1 = \{a_i,b_i,c_i,d_i,e_i,f_i : 1 \le i \le n\}$ and  $S_2= \{p_{ij},q_{ij},r_{ij},s_{ij},p_{ji},q_{ji},r_{ji},s_{ji} : (i,j) \in E\}$. From Proposition \ref{p1}, it is true that $ S_1 \subset S^\prime$, and also $ S_2 \subset S^\prime.$ Let $ S^*= S^\prime \setminus (S_1 \cup S_2 \cup \{t,u\}). $ Clearly, $ \vert S^* \vert \le n+k.$ Let $\vert S^* \cap X \vert = k^\prime $.  From claim \ref{xiyi}, clearly $ \vert S^* \cap (X \cup Y) \vert =2k^\prime.$ Since $ S^\prime $ is also a CDS of $ G^\prime$, $ \vert S^* \cap Z \vert = n-k^\prime $. Thus,
	$$2k^\prime+(n-k^\prime) \le n+k $$
	\begin{equation}\label{equa}
	k^\prime \le k
	\end{equation}
	\begin{claim}
		If $ VC =\{i : x_i, y_i \in S^\prime \},$ then $ VC $ forms a vertex cover in $ G. $ 
	\end{claim}	 
	\begin{proof}[Proof of claim]
		Let $ (i,j) \in E.$ From the construction of $ G $, it can be observed that there is no path from $ p_{ij} $ to $ b_i $ without  $ x_i $ or $ y_j $. Since $ S^\prime $ is connected, it should contain either $x_i$ or $y_j$ for each $ G_{ij} $. Similar argument can be made for each $ G_{ji}.$ Therefore, for each $ (i,j) \in E $ either $ x_i,y_i \in S^\prime $ or $ x_j,y_j \in S^\prime $. Hence, $ VC $ is a vertex cover in $ G $.		
	\end{proof}	 
	Therefore, from above claim and equation (\ref{equa}), clearly there exists a vertex cover of size at most $ k. $
\end{proof}
\subsection{Secure connected domination for chain graphs}
In this section, we propose a method to compute a minimum SCDS of a chain graph in linear time. A bipartite graph $ G=(X,Y,E) $ is called a \textit{chain graph} if the neighborhoods of the vertices of $ X $ form a \textit{chain}, that is, the vertices of $ X $ can be linearly ordered say, $ x_1,x_2,\ldots,x_p,$ such that $ N(x_1)\subseteq N(x_2) \subseteq\ldots \subseteq N(x_p).$ If a bipartite graph $ G=(X,Y,E) $ is a chain graph, then the neighborhoods of the vertices of $ Y $ also form a chain. An ordering $ \alpha=(x_1,x_2,\ldots,x_p,y_1,y_2,\ldots,y_q)$ of $ X\cup Y $ is called a \textit{chain ordering} if $N_G(x_1)\subseteq N_G(x_2) \subseteq\ldots \subseteq N_G(x_p)$ and $ N_G(y_1)\supseteq N_G(y_2) \supseteq\ldots \supseteq N_G(y_q)$. Every chain graph admits a chain ordering \cite{citekey}.
\begin{theorem}
	SCDM is linear time solvable for chain graphs.
\end{theorem}
\begin{proof}
	Let $ G=(X,Y,E) $ be a chain graph with chain ordering $ \{x_1,x_2,\ldots,x_p$, $y_1,y_2,\ldots,y_q\}$. If $ p=1$ or $q=1$ then $ G $ is a complete bipartite graph and clearly, $ \gamma_{sc}(G)=\vert X \cup Y\vert$. Otherwise, Let $ S= \{y_1,y_2,x_{p-1},x_p\} \cup P$, where $ P $ contains all the pendant vertices of $ G.$ It can be observed that for every vertex $ u \in V\setminus S $ there exists a vertex $ v \in S $ such that $ (S \setminus \{v\}) \cup \{u\}$ is a CDS of $ G$. Hence, $ S $ is a SCDS of $ G$ and $ \gamma_{sc}(G) \le \vert S \vert. $ \par
	Let $ S^\prime $ be any SCDS of $ G $, then we show that $ \vert S^\prime \vert$ $\ge$ $\vert S \vert.$ Note that if $ X \cap P \ne \emptyset$ ($ Y \cap P \ne \emptyset$) then $ y_1 $ ($ x_p $) is a support vertex. It is known that every SCDS contains all the pendant and support vertices of $ G.$ If $ P \ne \emptyset$ (Figure \ref{chaingraphs}(a) \& (b)) then clearly $ \vert S^\prime \vert$ $\ge$ $\vert S \vert.$ Otherwise, if $ \vert (S^\prime \cap Y) \vert < 2 $ then there exists a vertex $ u \in X \setminus S^\prime$ for which there is no vertex  $v \in S^\prime  $ such that $ (S^\prime\setminus\{v\}) \cup \{u\} $ is a CDS of $ G.$ Thus, $ \vert (S^\prime \cap Y) \vert \ge 2 $ (Figure \ref{chaingraphs}(c)). Similarly, $ \vert (S^\prime \cap X) \vert \ge 2.$ Hence, $ \vert S^\prime \vert$ $\ge$ $\vert S \vert.$ \par
	\begin{figure}
		\begin{center}
			\begin{tikzpicture}
			\node at (0,0){ \textbullet}; \node at (0,-0.2){$ y_1 $};
			\node at (1,0){ \textbullet}; \node at (1,-0.2){$ y_2 $};
			\node at (2,0){ \textbullet}; \node at (2,-0.2){$ y_3 $};
			
			\node at (0,2){\textbullet}; \node at (0,2.2){$ x_1 $};
			\node at (1,2){\textbullet}; \node at (1,2.2){$ x_2 $};
			\node at (2,2){\textbullet}; \node at (2,2.2){$ x_3 $};
			
			\draw (0,2)--(0,0)--(1,2);
			\draw (0,0)--(2,2)--(1,0);
			\draw (2,2)--(2,0);
			
			\node at (1,-0.8){(a)};

			\node at (4,0){ \textbullet}; \node at (4,-0.2){$ y_1 $};
			\node at (5,0){ \textbullet}; \node at (5,-0.2){$ y_2 $};
			\node at (6,0){ \textbullet}; \node at (6,-0.2){$ y_3 $};
			\node at (7,0){ \textbullet}; \node at (7,-0.2){$ y_4 $};
			
			\node at (4,2){\textbullet}; \node at (4,2.2){$ x_1 $};
			\node at (5,2){\textbullet}; \node at (5,2.2){$ x_2 $};
			\node at (6,2){\textbullet}; \node at (6,2.2){$ x_3 $};
			\node at (7,2){\textbullet}; \node at (7,2.2){$ x_4 $};
			
			\draw (4,2)--(4,0)--(5,2);
			\draw (7,2)--(4,0)--(6,2)--(5,0);
			\draw (6,0)--(7,2)--(7,0);
			\draw (4,2)--(5,0)--(5,2);
			\draw (5,0)--(7,2);
			
			\node at (5.5,-0.8){(b)};
			
			\node at (9,0){ \textbullet}; \node at (9,-0.2){$ y_1 $};
			\node at (10,0){ \textbullet}; \node at (10,-0.2){$ y_2 $};
			\node at (11,0){ \textbullet}; \node at (11,-0.2){$ y_3 $};
			\node at (12,0){ \textbullet}; \node at (12,-0.2){$ y_4 $};
			
			\node at (9,2){\textbullet}; \node at (9,2.2){$ x_1 $};
			\node at (10,2){\textbullet}; \node at (10,2.2){$ x_2 $};
			\node at (11,2){\textbullet}; \node at (11,2.2){$ x_3 $};
			\node at (12,2){\textbullet}; \node at (12,2.2){$ x_4 $};
			
			\draw (9,2)--(9,0)--(10,2);
			\draw (12,2)--(9,0)--(11,2)--(10,0);
			\draw (11,2)--(11,0)--(12,2)--(12,0);
			\draw (9,2)--(10,0)--(10,2);
			\draw (10,0)--(12,2);
			\draw (11,2)--(12,0);
			
			\node at (10.5,-0.8){(c)};
			\end{tikzpicture}
			\caption{SCDS in Chain graphs}\label{chaingraphs}
		\end{center}
	\end{figure}
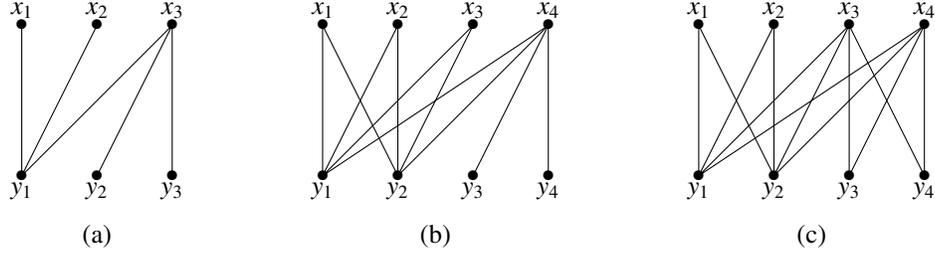
	In a chain graph $ G=(X,Y,E) $, a chain ordering and the set $ P $ of all pendant vertices can be computed in linear time \cite{chain}. Therefore, SCDM in chain graphs can be solved in linear time.   
\end{proof}
\subsection{Secure connected domination for bounded tree-width graphs}
Let $ G $ be a graph, $ T $ be a tree and $v\ $ be a family of vertex sets $ V_t \subseteq V (G) $ indexed by the vertices $ t $ of $ T $ . 
The pair $ (T, v\ ) $ is called a tree-decomposition of $ G $ if it satisfies the following three conditions:
(i) $ V(G)  =  \bigcup_{t \in V(T)} V_t $,
(ii) for every edge $ e \in E(G) $ there exists a $ t \in V(T)$ such that both ends of $ e $ lie in $ V_t $,
(iii) $V_{t_1} \cap   V_{t_3}  \subseteq V_{t_2}$ whenever $ t_1$, $t_2$, $t_3 \in V(T) $ and $ t_2 $ is on the path in $ T $ from $ t_1 $ to $ t_3 $.
The width of $ (T, v\ ) $ is the number $ max\{\vert V_t \vert-1 : t\in T \}$, and the tree-width $ tw(G) $ of $ G $ is the minimum width of any tree-decomposition of $ G $. By Courcelle's Thoerem, it is well known that every graph problem that can be described by counting monadic second-order logic (CMSOL) can be solved in linear-time in graphs of bounded tree-width, given a tree decomposition as input \cite{courc}. We show that SCDM problem can be expressed in CMSOL. 
\begin{theorem}[\textit{Courcelle's Theorem}](\cite{courc})\label{cmsol1}
	Let $ P $ be a graph property expressible in CMSOL and let $ k $ be
	a constant. Then, for any graph $ G $ of tree-width at most $ k $, it can be checked in
	linear-time whether $ G $ has property $ P $.
\end{theorem}
\begin{theorem}\label{cmsol2}
	Given a graph $ G $ and a positive integer $ k $, SCDM can be expressed in CMSOL.
\end{theorem}
\begin{proof}
	First, we present the CMSOL formula which expresses that the graph $ G $ has a dominating set of size at most $ k.$
{\small 	$$Dominating(S)= (\vert S \vert \le k)  \land (\forall p)((\exists q)(q\in S \land adj(p,q))) \lor (p\in S)$$}
where $ adj(p, q) $ is the binary adjacency relation which holds if and only if, $ p, q $ are two adjacent vertices of $ G.$
	$ Dominating(S) $ ensures that for every vertex $ p \in V $, either $ p\in S $ or $ p $ is adjacent to a vertex in $ S$ and the cardinality of $ S $ is at most $ k.$ 
	For a set $ S \subseteq V,$ the induced subgraph $ G[S] $ is disconnected if and only if the set $ S $ can be partitioned into two sets $ S_1 $ and $ S_2 $ such that there is no edge between a vertex in $ S_1 $ and a vertex in $ S_2 $. The CMSOL formula to express that the induced subgraph $ G[S] $ is connected as follows.
	{\small $$Connected(S)= \neg (\exists S_1, S_1 \subseteq S, \neg(\exists e \in E, \exists u \in S_1, \exists v \in S \setminus S_1, (inc(u,e) \land inc(v,e)) ))$$}
	where $ inc(v, e) $ is the binary incidence relation which hold if and only if edge $ e $ is incident to vertex $ v $ in $ G.$
	Now, by using the above two CMSOL formulas we can express SCDM in CMSOL formula as follows.\\
{\small 	$ SCDM(S)=Dominating(S) \land Connected(S) \land (\forall x)((x \in S) \lor$} \\{\small \hspace*{2cm}$((\exists y) (y \in S \land Dominating((S \setminus \{y\}) \cup \{x\}) \land Connected((S \setminus \{y\}) \cup \{x\}))))  $}\\
	Therefore, SCDM can be expressed in CMSOL.
\end{proof}
\noindent Now, the following result is immediate from Theorems \ref{cmsol1} and \ref{cmsol2}.
\begin{theorem}
		SCDM can be solvable in linear time for bounded tree-width graphs. 
\end{theorem}
\section{Approximation Results}
In this section, we obtain upper and lower bounds on the approximation ratio of the MSCDS problem. We also show that the MSCDS problem is APX-complete for graphs with maximum degree $ 4 $.
\subsection{Approximation Algorithm}
Here, we propose a $\Delta(G)+1$ approximation algorithm for the MSCDS problem. In this, we will make use of two known optimization problems, MINIMUM DOMINATION and MINIMUM CONNECTED DOMINATION. The following two theorems are the approximation results which have been obtained for these two problems.
\begin{theorem}(\cite{clrs})\label{appdom}
	The MINIMUM DOMINATION problem in a graph with maximum degree $ \Delta(G) $ can be approximated with an approximation ratio of $1+\ln (\Delta(G)+1).$
\end{theorem}
\begin{theorem}\label{appcdom}(\cite{guha})
	The MINIMUM CONNECTED DOMINATION problem in a graph with maximum degree $ \Delta(G) $ can be approximated with an approximation ratio of $ 3+\ln \Delta(G). $
\end{theorem}
By theorems \ref{appdom} and \ref{appcdom}, let us consider APPROX-DOM-SET and APPROX-CDS are the approximation algorithms to approximate the solutions for MINIMUM DOMINATION and MINIMUM CONNECTED DOMINATION with approximation ratios of $ 1+\ln (\Delta(G)+1)$ and $ 3+\ln \Delta(G)$ respectively.\par
Now, we propose an algorithm APPROX-SCDS to produce an approximate solution for the MSCDS problem. In APPROX-SCDS, first we compute CDS $ D_c $ of a given graph $ G$ using APPROX-CDS. Next, we obtain the induced subgraph $ G^\prime $ from $V \setminus D_c$. By using APPROX-DOM-SET, we compute dominating set $ D $ of $ G^\prime$. Let $ D_{sc}=D_c \cup D.$ It can be easily observed that for every vertex $ u \in V \setminus D_{sc} $ there exists a vertex $ v \in D $ such that $ (D_{sc}\setminus \{v\}) \cup \{u\} $ is a CDS of $ G.$ Therefore, $ D_{sc} $ is a SCDS of $ G. $

\renewcommand{\algorithmicrequire}{\textbf{Input:}}
\renewcommand{\algorithmicensure}{\textbf{Output:}}
\begin{algorithm}
	\caption{APPROX-SCDS($ G $)}\label{scdsalgo}
	\begin{algorithmic}[1]
		\Require{\mbox{A simple and undirected bipartite graph $G$ }}
		\Ensure{\mbox{A SCDS $ D_{sc} $ of $G$.}}\\
		$ D_c  \gets$ \textbf{APPROX-CDS} ($ G $) \\
		Let $ G^\prime= G[V\setminus D_c] $ \\
		$ D \gets$ \textbf{APPROX-DOM-SET} ($G^\prime$)  \\
		$ D_{sc}\gets D_c \cup D $ \\
		\Return $ D_{sc}. $
		
	\end{algorithmic}
\end{algorithm}
\begin{theorem}\label{appscdom}
	The MSCDS problem in a graph $ G $ with maximum degree $ \Delta(G)$ can be approximated with an approximation ratio of $(\Delta(G)+1).$ 
\end{theorem}
\begin{proof}
	To prove the theorem, we show that SCDS produced by our algorithm APPROX-SCDS, $ D_{sc}$, is of size at most  $(\Delta(G)+1) $ times of $ \gamma_{sc}(G)$, i.e., $$ \vert D_{sc}\vert  \le (\Delta(G)+1)\gamma_{sc}(G) $$  
From the algorithm,
	\begin{center}
	$\vert D_{sc} \vert = \vert D_c \cup D \vert$ \\
	\hspace*{1.76cm}	$ = \vert D_c \vert  + \vert D \vert \le n $\\
	\hspace*{1.95cm}  $ \le (\Delta(G)+1) \gamma(G) $  \\
	\hspace*{2.12cm}  $ \le (\Delta(G)+1) \gamma_{sc}(G)  $	
	\end{center}
\end{proof}
Since the MSCDS problem in a graph with maximum degree $ \Delta(G) $ admits an approximation algorithm that achieves the approximation ratio of $(\Delta(G)+1)$, we immediately have the following corollary of theorem \ref{appscdom}.
\begin{cor}
	The MSCDS problem is in the class of APX when the maximum degree $ \Delta(G) $ is fixed.
\end{cor}
\subsection{Lower bound on approximation ratio}
To obtain a lower bound, we provide an approximation preserving reduction from the MINIMUM DOMINATION problem, which has the following lower bound.
\begin{theorem}\label{dsinapp}\cite{inapprox}
	For a graph $ G=(V,E) $, the MINIMUM DOMINATION problem cannot be approximated within $ (1-\epsilon) \ln n$ for any $ \epsilon >0 $ unless NP $ \subseteq $ DTIME$( n ^{O(\log\log n)})$, where $ n=\vert V \vert$.  
\end{theorem}
\noindent The above result holds in bipartite and split graphs as well \cite{inapprox}.
\begin{theorem}\label{scdsinappthm}
	For a graph $ G=(V,E) $, the MSCDS problem cannot be approximated within $ (1-\epsilon) \ln \vert V \vert $ for any $ \epsilon >0 $ unless NP $ \subseteq $ DTIME$( \vert V \vert ^{O(\log\log \vert V \vert)}). $ 
\end{theorem}
\begin{proof}
	In order to prove the theorem, we propose the following approximation preserving reduction. Let $ G=(V,E) $, where $ V=\{v_1,v_2,\ldots,v_n\} $ be an instance of the MINIMUM DOMINATION problem. From this we construct an instance $ G^\prime=(V^\prime, E^\prime) $ of MSCDS, where $ V^\prime=V \cup \{w, z\} $, and $ E^\prime=E \cup \{(v_i,w): v_i \in V\}  \cup \{(w, z)\}$.  \par
	Let $ D^*$ be a minimum dominating set of a graph $ G $ and $ S^*$ be a minimum SCDS of a graph $ G^\prime.$ It can be observed from the reduction that by using any dominating set of $ G, $ a SCDS of $ G^\prime $ can be formed by adding $ w$ and  $z$ vertices to it. Hence $ \vert S^* \vert \le \vert D^* \vert + 2. $ \par
	Let algorithm $ A $ be a polynomial time approximation algorithm to solve the MSCDS problem on graph $ G^\prime $ with an approximation ratio $ \alpha=(1-\epsilon) \ln \vert V^\prime \vert $ for some fixed $ \epsilon>0. $ Let $ k $ be a fixed positive integer. Next, we propose the following algorithm, DOM-SET-APPROX to find a dominating set of a given graph $ G $.

\renewcommand{\algorithmicrequire}{\textbf{Input:}}
\renewcommand{\algorithmicensure}{\textbf{Output:}}
	\begin{algorithm}
		
		\caption{DOM-SET-APPROX($ G $)}
		\begin{algorithmic}[1]
			\Require{\mbox{A simple and undirected graph $G$ }}
			\Ensure{\mbox{A dominating set $ D $ of $G$.}}
			\If {there exists a dominating set $ D^\prime $ of size at most $ k $} \\
			{\hskip1.5em $ D  \gets D^\prime$} 
			\Else \\
			{\hskip1.5emConstruct the graph $ G^\prime $ \\
			\hskip1.5emCompute a SCDS $ S $ of $ G^\prime $ by using algorithm $ A $ \\
			\hskip1.5em$ D \gets S \cap V $} 
			\EndIf \\
			\Return $ D. $
		\end{algorithmic}
	\end{algorithm}
	The algorithm DOM-SET-APPROX runs in polynomial time. It can be noted that if $ D $ is a minimum dominating set of size at most $ k $, then it is optimal. Next, we analyze the case where $ D $ is not a minimum dominating set of size at most $ k.$ \par
	Let $ S^*$ be a minimum SCDS of $ G^\prime$, then $ \vert S^* \vert \ge k. $ Given a graph $ G $, DOM-SET-APPROX computes a dominating set of size $ \vert D \vert \le  \vert S \vert \le \alpha \vert S^* \vert \le \alpha (\vert D^* \vert + 2) = \alpha (1 + 2/\vert D^* \vert)\vert D^* \vert \le \alpha (1 + 2/k)\vert D^* \vert$. Therefore, DOM-SET-APPROX approximates a dominating set within a ratio $ \alpha (1 + 2/k).$ 
	If $ 2/k < \epsilon/2,$ then the approximation ratio $ \alpha (1 + 2/k) < (1-\epsilon) (1+\epsilon/2) \ln n= (1-\epsilon^\prime)\ln n$ where $\epsilon^\prime=\epsilon/2+\epsilon^2/2.$ \par
	By theorem \ref{dsinapp}, if the MINIMUM DOMINATION problem can be approximated within a ratio of $ (1-\epsilon^\prime)\ln n,$ then $ NP \subseteq DTIME(n^{O(\log  \log n)})$. Similarly, if the MSCDS problem can be approximated within a ratio of $ (1-\epsilon)\ln n,$ then $ NP \subseteq DTIME(n^{O(\log  \log n)})$. For large values of $ n $, $ \ln n \approxeq \ln(n+2) $, for a graph $ G^\prime=(V^\prime,E^\prime),$ where $ \vert V^\prime \vert = \vert V \vert +2,$ MSCDS problem cannot be approximated within a ratio of $ (1-\epsilon)\ln \vert V^\prime \vert$ unless $ NP \subseteq DTIME(\vert V ^\prime \vert^{O(\log \log \vert V^\prime \vert)} ). $
\end{proof}
\begin{theorem}
	For a bipartite graph $ G=(X,Y,E) $, the MSCDS problem cannot be approximated within $ (1-\epsilon) \ln n $ for any $ \epsilon >0 $ unless NP $ \subseteq $ DTIME$( n^{O(\log\log n)})$, where $ n=\vert X \cup Y \vert $.  
\end{theorem}
\begin{proof}
	In order to prove the theorem, we propose the following approximation preserving reduction. Consider $ G=(X,Y,E) $, where $ X=\{x_1,x_2,\ldots,x_p\} $ and $Y=\{y_1,y_2,\ldots,y_q\}$ be an instance of the MINIMUM DOMINATION problem. From this we construct an instance $ G^\prime=(X^\prime, Y^\prime,E^\prime) $ of MSCDS, where $ X^\prime=X \cup \{w_1, z_2\} $, $ Y^\prime=Y \cup \{z_1, w_2\} $ and $ E^\prime=E \cup \{(x_i,z_1): x_i \in X\}  \cup \{(y_i,w_1): y_i \in Y\} \cup \{(w_1,w_2), (z_1,z_2), (w_1,z_1)\}$. An example construction of graph $ G^\prime$ from a bipartite graph $ G=(X, Y, E) $ with $ X=\{x_1,x_2,x_3,x_4,x_5\},$ and $Y=\{y_1,y_2,y_3,y_4,y_5\} $  is illustrated in Figure \ref{bipartiteinappr}. \par
Let $ D^*$ be a minimum dominating set of a graph $ G $ and $ S^*$ be a minimum SCDS of a graph $ G^\prime.$ It can be observed from the reduction that by using any dominating set of $ G, $ a SCDS of $ G^\prime $ can be formed by adding $ \{w_1, w_2, z_1, z_2\} $ vertices to it. Hence, $ \vert S^* \vert \le \vert D^* \vert + 4. $ The rest of the proof is similar to the proof of theorem \ref{scdsinappthm}.
	\begin{figure}
		\begin{center}
			\begin{tikzpicture}[scale=1.1]
			\node at (0,0){\textbullet}; \node at (0.2,-0.2){$ x_5 $};
			\node at (0,1){\textbullet}; \node at (0,0.8){$ x_4 $};
			\node at (0,2){\textbullet}; \node at (0,1.8){$ x_3 $};
			\node at (0,3){\textbullet}; \node at (0,2.7){$ x_2 $};
			\node at (0,4){\textbullet}; \node at (0.1,3.7){$ x_1 $};
			
			\node at (3.2,0){\textbullet}; \node at (3.2,-0.2){$ y_5 $};
			\node at (3.2,1){\textbullet}; \node at (3.2,0.8){$ y_4 $};
			\node at (3.2,2){\textbullet}; \node at (3.2,1.8){$ y_3 $};
			\node at (3.2,3){\textbullet}; \node at (3.2,2.8){$ y_2 $};
			\node at (3.2,4){\textbullet}; \node at (3.2,3.7){$ y_1 $};
			
			\draw (0,4)--(3.2,4)--(0,3)--(3.2,3)--(0,4);
			\draw (0,4)--(3.2,4)--(0,3)--(3.2,3)--(0,4);
			\draw (0,2)--(3.2,1)--(0,1);
			\draw (0,2)--(3.2,3);
			\draw (0,2)--(3.2,0);
			\draw (3.2,2)--(0,3)--(3.2,1);
			\draw (0,0)--(3.2,0);
			
			\node at (-1,2){\textbullet};\node at (-1.3,1.8){$ z_1 $};
			\node at (4.2,2){\textbullet};\node at (4.45,1.8){$ w_1 $};
			
			\draw[dashed,thick] (0,0)--(-1,2);
			\draw[dashed,thick] (0,1)--(-1,2);
			\draw[dashed,thick] (0,2)--(-1,2);
			\draw[dashed,thick] (0,3)--(-1,2);
			\draw[dashed,thick] (0,4)--(-1,2);
			
			\draw[dashed,thick] (3.2,0)--(4.2,2);												
			\draw[dashed,thick] (3.2,1)--(4.2,2);												
			\draw[dashed,thick] (3.2,2)--(4.2,2);												
			\draw[dashed,thick] (3.2,3)--(4.2,2);												
			\draw[dashed,thick] (3.2,4)--(4.2,2);

			\node at (-2.1,2){\textbullet};\node at (-2,1.8){$ z_2 $};
			\draw[dashed,thick] (-1,2)--(-2.1,2);
			
			\node at (5.3,2){\textbullet};\node at (5.4,1.8){$ w_2 $};
			\draw[dashed,thick]	 (4.2,2)--(5.3,2);
			
			\draw[dashed,thick] (-1,2)..controls (-0.5,4.6)..(1.6,4.6);
			\draw[dashed,thick]   (1.6,4.6)..controls (3.7,4.6)..(4.2,2);
			\end{tikzpicture}
			\caption{Example construction of a graph $G^\prime$}\label{bipartiteinappr}			
		\end{center}
	\end{figure}
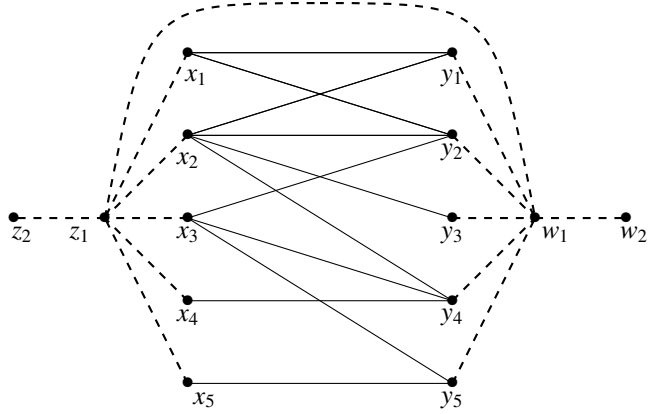
\end{proof}
\subsection{APX-completeness}
In this subsection, we prove that the MSCDS problem is APX-complete for graphs with maximum degree $ 4 $. This can be proved using an L-reduction, which is defined as follows.
\begin{definition}\textbf{(L-reduction)}
	Given two NP optimization problems $ F $ and $ G $ and a polynomial time transformation $ f $
	from instances of $ F $ to instances of $ G $, one can say that $ f $ is an \textit{L-reduction} if there exists positive constants $ \alpha $ and $ \beta $ such that for every instance $ x $ of $ F $ the following conditions are satisfied.
	\begin{enumerate}
		\item $ opt_G(f(x)) \le \alpha . opt_F(x) $.
		\item for every feasible solution $ y $ of $ f(x) $ with objective value $ m_G(f(x), y) = c_2 $ in polynomial time one can find a solution $ y^\prime $ of $ x $ with $ m_F (x, y^\prime) = c_1 $ such that $\vert  opt_F (x) - c_1 \vert \le  \beta \vert opt_G(f(x)) - c_2\vert. $
	\end{enumerate}
	Here, $ opt_F(x)$ represents the size of an optimal solution for an instance $ x $ of an NP optimization problem $ F$.  
\end{definition}
\noindent		An optimization problem $\pi $ is APX-complete if:
\begin{enumerate}
	\item $ \pi \in$ APX, and
	\item $ \pi \in$ APX-hard, i.e., there exists an L-reduction from some known APX-complete problem to $ \pi $.
\end{enumerate}
By theorem \ref{appscdom}, it is known that the MSCDS problem can be approximated within a constant factor for graphs with maximum degree $ 4 $. Thus, this problem is in APX for graphs with maximum degree $ 4 $. To show APX-hardness of MSCDS, we give an L-reduction from MINIMUM DOMINATING SET problem in graphs with maximum degree $ 3$ (DOM-$ 3 $) which has been proved as APX-complete \cite{dom3apx}.
\begin{theorem}
	The MSCDS problem is APX-complete for graphs with maximum degree $ 4.$
\end{theorem} 
\begin{proof}
	It is known that MSCDS is in APX. Given an instance $ G=(V,E)$ of DOM-$3$, where $ V=\{v_1,v_2,\ldots,v_n\} $, we construct an instance $ G^\prime=(V^\prime, E^\prime) $ of MSCDS as follows. Let $ X=\{x_1,x_2,\ldots,x_n\} $ and $ Y=\{y_1,y_2,\ldots,y_n\} $. In graph $ G^\prime,$ $ V^\prime= V \cup X \cup Y$ and $ E^\prime=E \cup \{(v_i,x_i),(x_i,y_i) : 1 \le i \le n\} \cup \{(x_i,x_{i+1}) : 1 \le i \le n-1\}.$ Note that $ G^\prime $ is a graph with maximum degree $ 4 $. An example construction of a graph $ G^\prime$ from a graph $ G$ is shown in Figure \ref{fig:APX}.  
	\begin{figure}
		\begin{center}
			\begin{tikzpicture}[scale=1.0]
			\draw[dotted] (1.25,2.9) ellipse (1.9cm and 2cm);
			\node at (1.2,4.5){$ G $};
			
			\node at (0,2){\textbullet}; \node at (0,1.7){$ v_2 $};			\node at (-1.5,2){\textbullet}; \node at (-1.6,1.7){$ x_2 $};
			\node at (-3,2){\textbullet}; \node at (-2.9,1.7){$ y_2 $};
			
			\node at (0,4){\textbullet}; \node at (-0.2,3.7){$ v_1 $};		
			\node at (-1.5,4){\textbullet}; \node at (-1.2,3.7){$ x_1 $};
			\node at (-3,4){\textbullet}; \node at (-3,3.7){$ y_1 $};

			\node at (4,2){\textbullet}; \node at (4,1.7){$ x_3 $};
			\node at (5.5,2){\textbullet}; \node at (5.5,1.7){$ y_3 $};
			\node at (2.5,2){\textbullet}; \node at (2.5,1.7){$ v_3 $};
			
			\node at (4,4){\textbullet}; \node at (4.3,3.7){$ x_4 $};
			\node at (5.5,4){\textbullet}; \node at (5.5,3.7){$ y_4 $};
			\node at (2.5,4){\textbullet}; \node at (2.69,3.7){$ v_4 $};
			
			\draw (2.5,4)--(2.5,2)--(0,2);
			\draw (2.5,4)--(0,4)--(0,2);
			\draw (0,2)--(2.5,4);

			\draw (0,4)--(-1.5,4)--(-3,4);
			\draw (0,2)--(-1.5,2)--(-3,2);
			\draw (2.5,2)--(4,2)--(5.5,2);
			\draw (2.5,4)--(4,4)--(5.5,4);
			
			\draw (-1.5,4)--(-1.5,2);
			\draw (4,2)--(4,4);
			\draw (-1.5,2)..controls (1.2,0.3)..(4,2);
			\end{tikzpicture}
			\caption{Construction of $ G^\prime$ from $ G $}
			\label{fig:APX}
		\end{center}
	\end{figure}
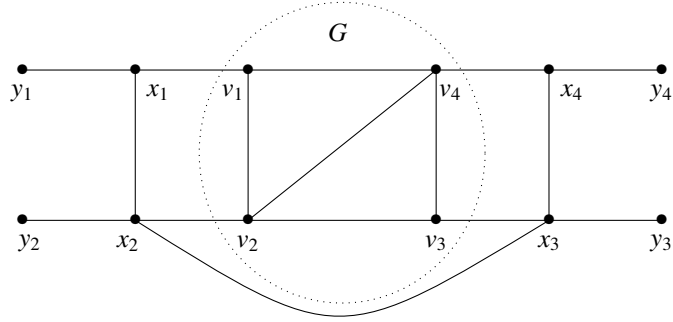
%
	\begin{claim}
		If $ D^*$ is a minimum dominating set of $ G$ and $ S^*$ is a minimum SCDS of $ G^\prime$ then $ \vert S^* \vert = \vert D^* \vert + 2n,$ where $ n=\vert V \vert.$
	\end{claim}
	\begin{proof}[Proof of claim]
		Suppose $ D^*$ is a minimum dominating set of $ G$, then $ D^* \cup X \cup Y $ is a SCDS of $ G^\prime.$ Further, if $ S^*$ is a minimum SCDS of $ G^\prime $, then it is clear that $ \vert S^* \vert \le \vert D ^* \vert + 2n.$ \par
		Next, we show that $ \vert S^* \vert \ge \vert D ^* \vert + 2n.$ Let $ S $ be any SCDS of $ G^\prime$. From Proposition \ref{p1}, it is clear that $ X \cup Y \subset S$, and no vertex $w \in X \cup Y $ is an $ S$-\textit{defender}. Therefore, for every vertex $ u \notin S$ there exists a vertex $ v \in S \cap V $ such that $ (S\setminus \{v\}) \cup \{u\} $ is a CDS of $ G^\prime.$ Hence $D=S \cap V $ is a dominating set of $ G$ and $\vert D \vert \ge \vert D^* \vert$ which implies $ \vert S \vert \ge \vert D^* \vert+2n.$ Since  $\vert S \vert \ge \vert S^* \vert$,  it is clear that $ \vert S^* \vert \ge \vert D ^* \vert + 2n.$ 
	\end{proof}
	Let $ D^*$ and $ S^* $ be a minimum dominating set and minimum SCDS of $ G $ and $ G^\prime$ respectively. It is known that for any graph $H$ with maximum degree $ \Delta(H) $, $ \gamma(H) \ge \frac{n}{\Delta(H)+1}$, where $ n=\vert V(H) \vert.$ Thus, $ \vert D^* \vert \ge \frac{n}{4}.$ From above claim it is evident that, $ \vert S^* \vert = \vert D ^* \vert + 2n \le \vert D ^* \vert+8\vert D ^* \vert = 9 \vert D ^* \vert.$ \par
	Now, consider a SCDS $ S $ of $ G^\prime$. Clearly, the set $ D=S \cap V$ is a dominating set of $ G.$ Therefore, $ \vert D \vert  \le \vert S \vert - 2n.$ Hence,  $ \vert D \vert-\vert D^* \vert  \le \vert S \vert - 2n-\vert D^*\vert=\vert S \vert -\vert S^* \vert.$ This proves that there is an L-reduction with $ \alpha=9 $ and $ \beta=1.$
\end{proof}
\section{Complexity difference in domination and secure connected domination}
Although secure connected domination is one of the several variants of domination problem, however they differ in computational complexity. In particular, there exist graph classes for which the first problem is polynomial-time solvable whereas the second problem is NP-complete and vice versa. Similar study has been made between domination and other domination parameters in \cite{henning, panda1,panda2}. \par

The DOMINATION problem is linear time solvable for doubly chordal graphs \cite{brandast}, but the SCDM problem is NP-complete for this class of graphs which is proved in section \ref{complexity}. Now, we construct a class of graphs in which the MSCDS problem can be solved trivially, whereas the DOMINATION problem is NP-complete.
\begin{definition}\textbf{(GC graph)}
	A graph is \textit{GC graph} if it can be constructed from a connected graph $ G=(V,E)$ where $ \vert V \vert=n,$ in the following way:\\[5pt]
		1. Create $ n $ complete graphs each with $ 3 $ vertices, such that $ i^{th} $ complete graph contains vertices $ \{a_i, b_i,c_i\} $.\\[3pt]
		2. Create $ n $ vertices, $ \{x_1,x_2,\ldots,x_n\}. $\\[3pt]
		3.  Add edges $ \{(x_i,v_i),(x_i,a_i): v_i \in V\}. $ 
\end{definition} 
\begin{figure}
	\begin{center}
		\begin{tikzpicture}[scale=1.0]
								\draw[dotted] (1.2,2.9) ellipse (1.8cm and 2cm);

		\node at (0,2){\textbullet}; \node at (0,1.8){$ v_2 $};			\node at (-1.5,2){\textbullet}; \node at (-1.5,1.8){$ x_2 $};
																		\node at (-3,2){\textbullet}; \node at (-2.9,1.8){$ a_2 $};
																		\node at (-4,2.5){\textbullet}; \node at (-4,2.75){$ b_2 $};
																		\node at (-4,1.5){\textbullet}; \node at (-4,1.3){$ c_2 $};

		\node at (0,4){\textbullet}; \node at (-0.2,3.8){$ v_1 $};			\node at (-1.5,4){\textbullet}; \node at (-1.5,3.8){$ x_1 $};
																		\node at (-3,4){\textbullet}; \node at (-3,3.8){$ a_1 $};
																		\node at (-4,4.5){\textbullet}; \node at (-4,4.7){$ b_1 $};
																		\node at (-4,3.5){\textbullet}; \node at (-4,3.3){$ c_1 $};

		\node at (4,2){\textbullet}; \node at (4,1.8){$ x_3 $};
																		\node at (5.5,2){\textbullet}; \node at (5.5,1.8){$ a_3 $};
																		\node at (6.5,2.5){\textbullet}; \node at (6.8,2.4){$ b_3 $};
																		\node at (6.5,1.5){\textbullet}; \node at (6.8,1.4){$ c_3 $};
																		\node at (2.5,2){\textbullet}; \node at (2.5,1.8){$ v_3 $};

		\node at (4,4){\textbullet}; \node at (4,3.8){$ x_4 $};
																		\node at (5.5,4){\textbullet}; \node at (5.5,3.8){$ a_4 $};
																		\node at (6.5,3.5){\textbullet}; \node at (6.8,3.4){$ c_4 $};
																		\node at (6.5,4.5){\textbullet}; \node at (6.8,4.4){$ b_4 $};
																		\node at (2.5,4){\textbullet}; \node at (2.5,3.8){$ v_4 $};
		
		\draw (0,4)--(2.5,2)--(0,2);
		\draw (2.5,4)--(0,4)--(0,2);
		\draw (0,2)--(2.5,4);

		\draw (0,4)--(-1.5,4)--(-3,4)--(-4,4.5);
		\draw (-3,4)--(-4,3.5)--(-4,4.5);
		
		\draw (0,2)--(-1.5,2)--(-3,2)--(-4,2.5)--(-4,1.5)--(-3,2);
		\draw (2.5,2)--(4,2)--(5.5,2)--(6.5,2.5)--(6.5,1.5)--(5.5,2);
		\draw (2.5,4)--(4,4)--(5.5,4)--(6.5,4.5)--(6.5,3.5)--(5.5,4);		

		\node at (1.2,4.6){$ G $};		
		\end{tikzpicture}
		\caption{GC graph construction}
		\label{gcgraph}
	\end{center}
\end{figure}
\begin{theorem}
	If $ G^\prime$ is a GC graph obtained from a graph $ G=(V,E) $ $ (\vert V \vert=n)$, then $ \gamma_{sc}(G^\prime)=4n$. 
\end{theorem}
\begin{proof}
	Let $ G^\prime=(V^\prime, E^\prime) $ be a GC graph. An example construction of GC graph is illustrated in Figure \ref{gcgraph}. Let $ S=V \cup \{x_1,x_2,\ldots,x_n\}\cup \{a_i,b_i : 1\le i\le n\}.$ It can be observed that $ S $ is a SCDS of $ G^\prime$ of size $ 4n$ and hence $ \gamma_{sc}(G^\prime) \le 4n.$ \par
	Let $ S $ be any SCDS in $ G^\prime $. It is known that every SCDS of a graph $ G $ is also a CDS of $ G $ and every CDS should contain all the cut-vertices of $ G $. Thus, it can be easily observed that for $ 1\le i \le n $, the vertices $ v_i, x_i$ and $a_i$ are cut-vertices in $ G^\prime$ and these vertices should be included in every SCDS of $ G^\prime $. Therefore, $ \vert S \vert \ge 3n. $ It can also be noted that these cannot defend any other vertex in $ G^\prime.$ Therefore, either $ b_i $ or $ c_i $, for each $ i $, where $ 1\le i\le n$ should be included in every SCDS of $ G^\prime$, and hence, $ \vert S \vert \ge 4n.$
\end{proof}
\begin{lemma}\label{difflemma}
	Let $ G^\prime$ be a GC graph constructed from a graph $ G=(V,E).$ Then $ G $ has a dominating set of size at most $ k $ if and only if $ G^\prime$ has a dominating set of size at most $ k+n.$
\end{lemma}
\begin{proof}
	Let $ A $ contain the degree $ 3 $ vertex from each copy of $ K_3.$ Suppose $ D $ is a dominating set of $ G $ of size at most $ k,$ then it is clear that $ D \cup A $ is a dominating set of $ G^\prime $ of size at most $ k+n.$ \par
	Conversely, suppose $ D^\prime $ is a dominating set of $ G^\prime$ of size $ k+n.$ Then at least one vertex from each $ k_3 $ must be included in $ D^\prime.$ 
Let $ D^{\prime\prime}$ be the set formed by replacing all $ x_i$'s in $ D^\prime $ with corresponding $ v_i $'s. Clearly, $ D^{\prime\prime} $ is a dominating set of size at most $k$ in $ G$.
\end{proof}
\noindent The following result is well known for the DOMINATION problem.
\begin{theorem}(\cite{garey})
	The DOMINATION problem is NP-complete for general graphs.
\end{theorem}
\begin{theorem}
	The DOMINATION problem is NP-complete for GC graphs.
\end{theorem}
\begin{proof}
	The proof directly follows from above theorem and lemma \ref{difflemma}.
\end{proof}
\noindent It is identified that the two problems, DOMINATION and SCDM are not equivalent in aspects of computational complexity. For example, when the input graph is either doubly chordal or a GC graph then complexities differ. Thus, there is a scope to study each of these problems on its own for particular graph classes. \vspace{0.2cm}
\begin{center}
	\textbf{Acknowledgement}
\end{center}
The authors are grateful to the anonymous reviewers for their valuable comments and suggestions, which result in the present version of the paper.

\end{document}